\documentclass[11pt, onecolumn]{article}
\usepackage[top=1in, bottom=1in, left=1.25in, right=1.25in]{geometry}

\usepackage{amsfonts}
\usepackage{amsmath,amssymb}
\usepackage{graphicx}
\usepackage{color,soul}
\usepackage{cite}
\usepackage{hyperref}
\usepackage{algorithm,algorithmic}
\usepackage{bm}
\usepackage{booktabs}
\usepackage{flushend}
\usepackage{tikz}
\usetikzlibrary{arrows}
\usepackage{subfigure}

\usepackage[amsmath,thmmarks]{ntheorem}
\usepackage{theorem}

\newtheorem{lem}{Lemma}

\newtheorem{rem}{Remark}
\newtheorem{thm}{Theorem}
\newtheorem{conj}{Conjecture}

\theoremheaderfont{\sc}\theorembodyfont{\upshape}
\theoremstyle{nonumberplain}
\theoremseparator{}
\theoremsymbol{\rule{1ex}{1ex}}
\newtheorem{proof}{Proof}

\newcommand{\vh}{{\mathbf{h}}}

\newcommand{\vq}{{\mathbf{q}}}

\newcommand{\vt}{{\mathbf{t}}}

\newcommand{\vA}{{\mathbf{A}}}

\newcommand{\vD}{{\mathbf{D}}}

\newcommand{\vH}{{\mathbf{H}}}
\newcommand{\vI}{{\mathbf{I}}}

\newcommand{\vQ}{{\mathbf{Q}}}

\newcommand{\vU}{{\mathbf{U}}}

\graphicspath{{../}}

\begin{document}

\title{On the Outage Probability Conjecture for MIMO Channels}

\author{Gen Li, Jingkai Yan, and Yuantao Gu%
\thanks{
The authors are with the Department of Electronic Engineering and Tsinghua National Laboratory for Information Science and Technology (TNList), Tsinghua University, Beijing 100084, CHINA.
The corresponding author of this paper is Yuantao Gu (gyt@tsinghua.edu.cn).
}}
\date{}

\maketitle

\begin{abstract}
  Multiple-Input-Multiple-Output (MIMO) communication systems have seen wide application due to its performance benefits such as multiplexing gain. For MIMO systems with non-ergodic Gaussian channel, a conjecture regarding its outage probability has been proposed by Telatar in \cite{telatar1999capacity}. This conjecture has been proved for the special single-output case, and is in general assumed to be true. In this work, we analyze the special Two-Input-Multiple-Output (TIMO) case theoretically, and provide a counter-example to the conjecture. The counter-example is verified both theoretically and by numerical experiments. We also present a theoretical analysis for general MIMO case, including a method for calculation. This result rejects the decades-long conjecture and provides interesting insight into the symmetry of MIMO systems.

\end{abstract}

\section{Introduction}

The method of Multi-Input-Multi-Output (MIMO) has seen wide application since its pioneering proposal in the 1990s \cite{paulraj1994increasing,raleigh1998spatio,foschini1996layered}. The use of MIMO significantly improves data rates and allows for multi-user communication \cite{zheng2003diversity,andrews2007fundamentals}, and has currently been adopted as standard by WLAN, WiMAX and LTE networks \cite{xiao2005ieee,damnjanovic2011survey,andrews2007fundamentals}. Furthermore, MIMO plays an essential part in the future generations of communication systems, as it is crucial to the currently developing 5G \cite{boccardi2014five,jungnickel2014role} and the popular mmWave MIMO systems \cite{ayach2014spatially,heath2015overview}. Reviews on the important role of MIMO in contemporary communication are provided in \cite{bj2015massive,heath2015overview}.


The channel capacity of MIMO system has been studied extensively by researchers. Telatar is believed to be the first to derive theoretical results regarding MIMO capacity \cite{telatar1999capacity}. In that work he proposed expressions of capacity of mean (ergodic) Gaussian channel with Rayleigh fading. A number of later works also address this problem. \cite{marzetta1999capacity} discussed the capacity and obtained insightful results on parameter setting. \cite{loyka2001channel} discussed the capacity when correlation exists in real-world circumstances. The case when interference occurs with multiple users sharing the channel is addressed by \cite{blum2003mimo}.


Despite the abundance of existing work, the vast majority is only focusing on the problem for mean (ergodic) Gaussian channel. For non-ergodic channels, on the other hand, previous analysis techniques for calculating channel capacity would fail because the capacity is in general not equal to the maximum mutual information \cite{telatar1999capacity}. In the work of \cite{telatar1999capacity}, the author provided an alternative analysis using the outage probability. \cite{oyman2006non} studied the power-bandwidth tradeoff for non-ergodic channels. \cite{letzepis2009outage} studied the outage probability for block-fading non-ergodic channels. A conjecture regarding optimum condition for the minimum outage probability was proposed in \cite{telatar1999capacity}. The conjecture has been widely assumed to be true \cite{prasad2003outage, abbe2012proof}, and a special case of Multiple-Input-Single-Output (MISO) was proved by \cite{abbe2012proof}.

In this work, we provide a new analysis technique to the outage probability conjecture, particularly in the case of Two-Input-Multiple-Output (TIMO), and prove that the conjecture can actually be false. We provide a concrete counter-example, verified both by theoretical calculation and numerical simulation. We also present an analysis for the general MIMO setting, and provide a method to calculate the derivatives involved, as well as for a particular type of power distribution.

The rest of this paper is organized as follows. Section II reviews the Outage Probability Conjecture and the proven case of MISO. Theoretical analysis of the TIMO case and the counter-example is given in Section III, which is the main contribution of our work. Section IV discusses the general MIMO case. Numerical experiments are conducted in Section V. Section VI offers concluding remarks.

\emph{Notations}:
$\vH^{\rm T}$ and $\vH^*$ denote respectively the transpose and conjugate transpose of $\vH$.
$\#\mathcal{S}$ denotes the number of elements in set $\mathcal S$.

\section{The Outage Probability Conjecture}

\subsection{The Conjecture}

Consider a single-user MIMO Gaussian channel. Assume the number of transmitting antennas and receiving antennas is $t$ and $r$, respectively. The channel model can be described as
\begin{equation}
 {\bf y} = {\bf H} {\bf x} + {\bf n},
\end{equation}
where ${\bf x}\in \mathbb{C}^t$, ${\bf y}\in \mathbb{C}^r$, and ${\bf n}\in \mathbb{C}^r$ denotes
the transmitted vector, the received vector, and the additive noise, respectively.
The channel is denoted by ${\bf H} \in \mathbb{C}^{r\times t}$,
where each entry is \emph{i.i.d.} circularly symmetric complex Gaussian variable with
$$
\mathbb{E}\left|H_{i,j}\right|^2 = 1.
$$

In the non-ergodic case, the matrix $\bf H$ is random but is held fix once it is chosen. In this case, we consider the outage probability for evaluation of the channel. Let $R$ be the data rate and $P$ be the signal power constraint. The outage probability $P_\mathrm{out}(R,P)$ is defined as follows:
\begin{equation} \label{eqn_outprob_orig}
P_\mathrm{out}(R,P) := \inf_{\substack{\vQ: \vQ \geq 0,\\ \text{tr}(\vQ) \leq P}} \mathbb{P}[ \log \det (\vI_r + \vH \vQ \vH^*) < R].
\end{equation}
In Telatar's words, $P_\mathrm{out}(R,P)$ is a probability such that

\begin{quote}
\emph{For any rate less than $R$ and any $\delta$ there exists a code satisfying the power constraint $P$ for which the error probability is less than $\delta$ for all but a set of $\vH$ whose total probability is less than $P_\mathrm{out}(R,P)$.}
\end{quote}

The conjecture is stated as the following \cite{telatar1999capacity}:

\begin{conj} [Outage Probability Conjecture]
  The matrices $\vQ$ that yield the infimum in \eqref{eqn_outprob} has eigenvalue decomposition $\vU \vD \vU^*$, where $\vU$ is $t\times t$ unitary and $\vD$ has the form of
  \begin{equation} \label{eqn_conj}
    \vD = \frac{P}{k} \mathrm{diag} \left( \underbrace{1,\dots,1}_{k}, \underbrace{0,\dots,0}_{t-k} \right),
  \end{equation}
  where $1\le k\le t$ is some integer.
\end{conj}

The intuition behind this conjecture is the symmetry between the transmitters, where they are either working at some same power ${P}/{k}$ or not working. This seems very plausible and therefore has been generally accepted as true.

While the above gives the original form of the Conjecture, in actuality we need only to consider the case of $\mathrm{tr}(\vQ)=P$, as the outage probability becomes smaller as maximum allowed power increases \cite{abbe2012proof}. Therefore the model can be expressed as
\begin{equation} \label{eqn_outprob}
P_\mathrm{out}(R,P) := \inf_{\substack{\vQ: \vQ \ge 0,\\ \text{tr}(\vQ) = P}} \mathbb{P}[ \log \det (\vI_r + \vH \vQ \vH^*) < R].
\end{equation}

\subsection{The Proven MISO Case}

In the work of Abbe et al \cite{abbe2012proof}, the authors in particular focused on the MISO case and offered a proof for the Conjecture. For the MISO case where $r=1$, we are able to remove the determinant in \eqref{eqn_outprob}, leading to the following conjecture:
\begin{conj} [Outage Probability Conjecture for MISO]
Let
$$
\mathcal{D}(t):=\{\vQ\in \mathbb{C}^{t\times t} \ | \ \vQ\ge 0, \mathrm{tr}(\vQ)\le 1\}
$$
and $\left(H_i\right)_{1\le i\le t}$ $\stackrel{i.i.d}{\sim} \mathcal{N}_\mathbb{C}(0,1)$. For all $x\in \mathbb{R}_+$, there exists $k\in \{1,\dots,t\}$ such that
\begin{align}
    \nonumber
    \arg\min_{\vQ\in \mathcal{D}(t)} \mathbb{P}\{\vH\vQ\vH^*\le x\} = 
     \left\{ \vU \mathrm{diag}\left( \underbrace{ \frac{1}{k},\dots,\frac{1}{k}}_{k} , \underbrace{0,\dots,0}_{t-k}\right) \vU^* : \vU\in \mathcal{U}(t) \right\}, \label{eqn_conj2}
\end{align}
where $\mathcal{U}(t)$ denotes the group of $t\times t$ unitary matrices.
\end{conj}

After \emph{transforming the conjecture into the problem of Gaussian quadratic forms having largest tail probability corresponding to such diagonal matrices},
the proof of MISO case was completed. For specific details please refer to \cite{abbe2012proof}.

\section{Analysis of TIMO Case and Counter-Example}

In this section, we present our analysis of the Two-Input-Multiple-Output case of the Conjecture. A counter-example is shown and verified.

\subsection{Main Results}

Addressing the general MIMO version of the Outage Probability Conjecture is rather hard, where the major difficulty lies in handling the determinant on the right hand side of \eqref{eqn_outprob}. This also partially explains why the MISO case is comparatively earlier to solve than the general case.

For the TIMO case, we employ another approach to the simplification of the determinant. In this case we can write
\begin{align*}
{\vQ} &= {\rm diag}(q_1,q_2),\\
\vH &= [\vh_1, \vh_2].
\end{align*}
According to Sylvester's determinant theorem \cite{harville1997matrix}, we have
\begin{align}\label{firstderivation}
\det (\vI_r + \vH \vQ \vH^*) 
=& \det (\vI_r + \vH^* \vH \vQ) \nonumber\\
=& 1 \!+\! q_1|\vh_1|^2 \!+\! q_2|\vh_2|^2 \!+\! q_1q_2|\vh_1|^2|\vh_2|^2\!\left(\!1 \!-\! \frac{|\vh_1^*\vh_2|^2}{|\vh_1|^2|\vh_2|^2}\!\right),
\end{align}

By introducing three random variables $S, T$, and $\rho$, we observe the important fact that
\begin{align}
\nonumber
S := q_1|\vh_1|^2 &\sim \frac{q_1}{2} \chi^2_{2r}, \\
\nonumber
T := q_2|\vh_2|^2 &\sim \frac{q_2}{2} \chi^2_{2r}, \\
\nonumber
1 - \rho := \frac{|\vh_1^*\vh_2|^2}{|\vh_1|^2|\vh_2|^2} &\sim {\rm Beta}(1, r-1),
\end{align}
where the probability density $f(x;k)$ for Chi-squared distribution $\chi_k^2$ and $f(x;\alpha,\beta)$ for Beta distribution ${\rm Beta}(\alpha,\beta)$ are, respectively,
\begin{align*}
f(x;k) =& \left\{
    \begin{aligned}
    \frac{x^{\frac{k}{2}-1}\text{e}^{-\frac{x}{2}}}{2^{\frac{k}{2}}\Gamma{\frac{k}{2}}},\quad &x>0; \\
    0, \quad &\text{otherwise}.
    \end{aligned}
\right.\\
f(x;\alpha,\beta) =& \left\{
    \begin{aligned}
    \frac{\Gamma(\alpha+\beta)}{\Gamma(\alpha)\Gamma(\beta)} x^{\alpha-1} (1-x)^{\beta-1} ,\quad &0<x<1; \\
    0, \quad &\text{otherwise}.
    \end{aligned}
\right.
\end{align*}
Notice that these three random variables are independent of each other, as $\frac{|\vh_1^*\vh_2|^2}{|\vh_1|^2|\vh_2|^2}$ can be interpreted as the angle between vectors $\vh_1$ and $\vh_2$, and is therefore independent of the vector lengths.

Based on \eqref{firstderivation} and above interpretation, we could calculate the right hand side of \eqref{eqn_outprob} by a probability integral.
When $q_1, q_2 \ne 0$, the integral is
\begin{align}
&\mathbb{P}[\log \det (\vI_r + \vH \vQ \vH^*) < R]\nonumber\\
=&\mathbb{P}[1+S+T+ST\rho < {\rm e}^R]\nonumber\\
=& \int_0^1 (r-1)\rho^{r-2}{\rm d}\rho 
 \int_0^{{\rm e}^R-1} \frac{s^{r-1}{\rm e}^{-\frac{s}{q_1}}}{(r-1)!q_1^r} \int_0^{\frac{{{\rm e}^R-1} -s}{1+\rho s}} \frac{t^{r-1}{\rm e}^{-\frac{t}{q_2}}}{(r-1)!q_2^r}{\rm d}t {\rm d}s \label{eqn_solution}\\
=:& f(q_1,q_2) \nonumber
\end{align}
When $q_1$ or $q_2$ equals zero, assuming that $q_1 = P, q_2 = 0$ without loss of generality, the integral is
\begin{equation}
\mathbb{P}[\log \det (\vI_r + \vH \vQ \vH^*) < R] = \int_0^{{\rm e}^R-1} \frac{s^{r-1}{\rm e}^{-\frac{s}{P}}}{(r-1)!P^r}{\rm d}s, \label{eqn_solution_0}
\end{equation}
which can also be derived indirectly from \eqref{eqn_solution} by taking the limit $q_1 = P$ and $q_2 \to 0$. This indicates that there is no need to make explicit distinction between the two cases, and we shall use the nonzero case in the following discussion when no ambiguity is possible.

According to the above analysis, the original conjecture is equivalent to the following:
\begin{conj} [Outage Probability Conjecture for TIMO]\label{conjTIMO}
Let
$$
\mathcal{D} = \{(q_1,q_2)\in [0,P]^2 \ | \ q_1+q_2 = P\}.
$$
The pairs of $(q_1,q_2)$ that minimizes $f(q_1,q_2)$ in domain $\mathcal D$ fall in the set
$$
\left\{ (P,0), \left(\frac{P}{2},\frac{P}{2}\right), (0,P) \right\}.
$$
\end{conj}

The right hand side expression in \eqref{eqn_solution}, i.e. $f = f(q_1,q_2)$, is symmetric in the sense of $f(q_1,q_2)=f(q_2,q_1)$. It should be noted that while $f$ has the form of a bivariate function, it is intrinsically univariate in our problem \eqref{eqn_outprob} due to the relation $q_1+q_2=P$.

Applying derivatives and exploiting the symmetry between $q_1$ and $q_2$, the following result can be directly obtained as a sufficient condition for rejecting the Outage Probability Conjecture:
\begin{thm} \label{thm_1}
  If there exist positive integers $m, n$, such that
  \begin{align}
    \left. \frac{\mathrm{d}^i f}{\mathrm{d} q_1^i} \right|_{q_1 = 0} &= 0, \quad 1\le i\le m-1, \label{eqn_thm_1}\\
    \left. \frac{\mathrm{d}^m f}{\mathrm{d} q_1^m} \right|_{q_1 = 0} &< 0, \label{eqn_thm_2}\\
    \left. \frac{\mathrm{d}^j f}{\mathrm{d} q_1^j} \right|_{q_1 = \frac{P}{2}} &= 0, \quad 1\le j\le n-1,\label{eqn_thm_3}\\
    \left. \frac{\mathrm{d}^n f}{\mathrm{d} q_1^n} \right|_{q_1 = \frac{P}{2}} &< 0, \label{eqn_thm_4}
  \end{align}
  then the Conjecture \ref{conjTIMO} does not hold.
\end{thm}

\begin{proof}
From the definition of derivatives, the above conditions imply that $f(q,P-q)$ is monotone decreasing in some interval $[0,\epsilon_1)$ and $[\frac{P}{2},\frac{P}{2}+\epsilon_2)$, where $\epsilon_1$ and $\epsilon_2$ are sufficiently small. By symmetry $f$ is monotone increasing on $(\frac{P}{2}-\epsilon_2,\frac{P}{2}]$. Therefore from the continuity of $f$ there exist some $q_m$ with $\epsilon_1 < q_m < \frac{P}{2}-\epsilon_2$ such that $f(q_m,P-q_m)<\min(f(0,P),f(\frac{P}{2},\frac{P}{2}))$.
\end{proof}

\begin{rem} \label{rem_1}
Observe the fact that an analytic function will not take the value of 0 over a set with measure greater than zero, unless it is constant zero. Therefore for general values of $R$ and $P$, the first and the second order derivatives will not be constantly 0, and we need not consider the cases $m>1$ and $n>2$ in general. Note that when $q_1=q_2=\frac{P}{2}$, the first derivative of $f$ to $q_1, q_2$ at this point is constantly zero. Therefore $n=2$ has to be considered.
\end{rem}

\subsection{A Counter-Example}

From the results above, it becomes theoretically straightforward to identify counter-examples by examining conditions \eqref{eqn_thm_1}-\eqref{eqn_thm_4}. According to Remark \ref{rem_1}, we only need to check the first and the second order derivatives of expression \eqref{eqn_solution}, the derivation of which is included in Appendix.

As a specific instance, take the set of parameters
$$
t=2, \quad r=2, \quad P=0.5,  \quad R=\ln 3.
$$
A direct verification of the first and the second order derivatives using results in Appendix shows that
\begin{align*}
  \left. \frac{\mathrm{d} f}{\mathrm{d} q_1} \right|_{q_1 = 0} &= 0,\\
  \left. \frac{\mathrm{d}^2 f}{\mathrm{d} q_1^2} \right|_{q_1 = 0} &= -\frac{8}{\mathrm{e}^4} <0, \\
  \left. \frac{\mathrm{d} f}{\mathrm{d} q_1} \right|_{q_1 = \frac{P}{2}} &= 0,\\
  \left. \frac{\mathrm{d}^2 f}{\mathrm{d} q_1^2} \right|_{q_1 = \frac{P}{2}} &= \int_0^1 \int_0^2 \gamma(t,\rho) \mathrm{d}s \mathrm{d}\rho \approx -0.1014 <0,
\end{align*}
where
\begin{align*}
\gamma(s,\rho) &= 4^6 s t^2 \mathrm{e}^{-4\left(s+t\right)} \left( 2-8t+8s \right),\\
t &=\frac{2-s}{1+\rho s}.
\end{align*}
Therefore there would exist some $q_m$ with $0<q_m<\frac{P}{2}$ such that $(q_1,q_2)=(q_m,P-q_m)$ yields a smaller outage probability than as predicted by the Outage Probability Conjecture.

A numerical verification of the counter-example is demonstrated in Fig. \ref{fig_counter}.
\begin{figure}[!t]
  \centering
  \includegraphics[width=2.8in]{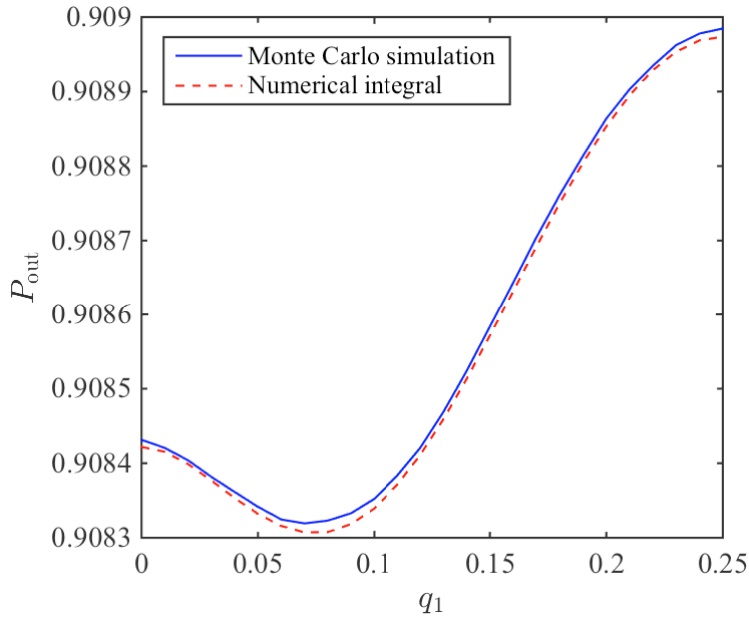}
  \caption{Outage Probability is plotted as a function of $q_1$ under parameters $t=2, r=2, P=0.5, R=\ln 3$.
  Both Monte Carlo simulation and numerical integration are included.
  It can be from the curve that the minimum is attained within the interval $(0.05,0.1)$, instead of at 0 or ${P}/{2}$ as predicted by the Conjecture.}
  \label{fig_counter}
\end{figure}

\begin{rem}
  The above counter-example suggests that the seemingly intuitive symmetry in the Outage Probability Conjecture is actually misleading. One possible explanation is that the expression \eqref{eqn_outprob} in calculating outage probability is non-convex, and therefore does not enjoy many of the promising properties of convex functions.
\end{rem}

\section{Discussion of the General MIMO Case}

In this section we provide some discussions about dealing with the general MIMO case with arbitrary $t$. We first derive a generalization of the previous sufficient condition, and then provide a method for the calculation of derivatives as well as for direct calculation of the outage probability under a certain type of power distribution.

\subsection{Extending the Sufficient Condition}

We can generalize our main result into general cases with arbitrary $t$.
Denote the power distribution vector by
$$
\tilde{\vq}(k,t) = \frac{P}{k}\left[ \underbrace{1,\dots,1}_{k}, \underbrace{0,\dots,0}_{t-k} \right]^{\rm T},  \quad 1\le k\le t,
$$
where $k$ nonzero transmitters with uniform power are present among a total of $t$. The first $k$ are set as nonzero without loss of generality. According to the Conjecture, the minimizer of $f(\vq)$ is among these vectors.

Similar to the intuition of Theorem \ref{thm_1}, the following result can be obtained:
\begin{thm} \label{thm_2}
For arbitrary $t$, the Outage Probability Conjecture does not hold,
if for all integers $k$ with $1\le k\le t$, at least one of the following conditions is satisfied :
  \begin{align}
    \left. \frac{\partial f}{\partial q_i} \right|_{\vq=\tilde{\vq}(k,t)} - \left. \frac{\partial f}{\partial q_l} \right|_{\vq=\tilde{\vq}(k,t)} &> 0, \label{eqn_anyT1} \\
    \left. \frac{\partial^2 f}{\partial q_i^2} \right|_{\vq=\tilde{\vq}(k,t)} - \left. \frac{\partial^2 f}{ \partial q_i \partial q_j} \right|_{\vq=\tilde{\vq}(k,t)} &> 0, \label{eqn_anyT2}
  \end{align}
where $i$ and $j$ are two different indices such that $q_i=q_j={P}/{k}$, and $l$ is an index such that $q_l=0$.
In particular, when $k=t$ condition \eqref{eqn_anyT1} is omitted, and when $k=1$ condition \eqref{eqn_anyT2} is omitted. \footnote{The specific choice of the indices does not matter due to symmetry.}
\end{thm}
\begin{proof}
In order to prove that the Conjecture fails, it suffices to show that at every $\tilde{\vq}(k,t)$, some slight variation on $\tilde{\vq}$ within feasible region would result in a decrease in function value $f$, so that $\tilde{\vq}$ cannot be a local minimum point. For this variation, we in particular look at two special cases, $\tilde{\vq}'(k,t)$ and $\tilde{\vq}''(k,t)$, where only two entries are modified, namely the following:
\begin{align}
\nonumber
\tilde{\vq}'(k,t) &= \left[ \frac{P}{k}-\epsilon, \underbrace{\frac{P}{k},\dots,\frac{P}{k}}_{k-1}, \epsilon, \underbrace{0,\dots,0}_{t-k-1} \right]^{\rm T}, \\
\nonumber
\tilde{\vq}''(k,t) &= \left[ \frac{P}{k}-\epsilon, \frac{P}{k}+\epsilon, \underbrace{\frac{P}{k},\dots,\frac{P}{k}}_{k-2}, \underbrace{0,\dots,0}_{t-k} \right]^{\rm T},
\end{align}
where $\epsilon>0$ is sufficiently small. As a matter of fact, inspecting the above two cases is both sufficient and necessary for showing that $\tilde{\vq}(k,t)$ is not the local minimum. Take any variation $\tilde{\vq}+\epsilon \vt$ at $\tilde{\vq}$ with $\left|\vt\right|=1$, the increment $\epsilon \vt$ can be decomposed into some combinations of the above two patterns, by exploiting the symmetry between variables.

Now we observe that $f(\vq'(k,t))<f(\vq(k,t))$ is equivalent to \eqref{eqn_anyT1}, and $f(\vq''(k,t))<f(\vq(k,t))$ is equivalent to \eqref{eqn_anyT2}. Therefore when for every $\tilde{\vq}(k,t)$ at least one of these two inequalities holds, it is implied that $\tilde{\vq}(k,t)$ is not the local minimum and thus not the global minimum. This contradicts the Conjecture.
\end{proof}

\subsection{Special Power Distribution and Calculating Derivatives}

In this part we present an alternative method for evaluating the expressions in Theorem \ref{thm_2}, as well as calculating the outage probability \eqref{eqn_outprob} directly for arbitrary $t$ under a particular constraint on the power distribution matrix $\vQ$.

Let $\vQ =\mathrm{diag}(q_1,\dots,q_t)$.
We focus on cases where there exists $q_0 > 0$ such that
$$
\#\left\{ q_i \ |\ q_i\notin \{0,q_0\}, 1\le i\le t \right\} \le 2.
$$
In other words, at most two non-zero diagonal entries deviate from a common value $q_0$. These two possibly deviant values are denoted by $q_a$ and $q_b$. We further assume that $q_a,q_b > 0$, because when they take the value of 0, we can always reduce $k$ or arrive at the trivial case of $k=1$. Also, this particular type of $\vQ$ is sufficient to describe any power distribution for $t=3$ cases.

Without loss of generality, we can assume that
$$
\vQ = \mathrm{diag} \left( \underbrace{q_0,\dots,q_0}_{k-2}, q_a, q_b, \underbrace{0,\dots,0}_{t-k} \right).
$$
These transmitters can be divided into three groups: the uniform, the deviant, and the non-functioning. We define the sub-matrices $\vQ_1$ and $\vQ_2$ as the following:
\begin{align}
  \vQ_1 &= \mathrm{diag} \left( \underbrace{q_0,\dots,q_0}_{k-2} \right) = q_0 \vI_{k-2}. \\
  \vQ_2 &= \mathrm{diag} \left( q_a, q_b \right) =
  \begin{bmatrix}
    q_a & 0 \\
    0 & q_b
  \end{bmatrix}.
\end{align}
The Gaussian matrix $\vH$ can be divided correspondingly as
$$
\vH = [\vH_1, \vH_2, \vH_0],
$$
where $\vH_1 \in \mathbb{C}^{r\times (k-2)}, \vH_2 \in \mathbb{C}^{r\times 2}, \vH_0 \in \mathbb{C}^{r\times (t-k)}$.

The following lemma is required for our discussion.
\begin{lem}\cite{ratnarajah2005eigenvalues}\label{lemma1}
For a standard complex Gaussian matrix $\vA \in \mathbb{C}^{n\times m}$,
the distribution density of the eigenvalues $\bm{\lambda} = [\lambda_1,\dots,\lambda_m]^{\rm T}$ of Hermitian matrix $\vA^*\vA$ was derived as the following:
\begin{align*}
  \varphi_{m,n}(\bm{\lambda}) = \frac{2^{-mn}\pi^{m(m-1)}}{\Gamma_m(n)\Gamma_m(m)} 
  \exp\left(-\frac{1}{2}\sum_{i=1}^m \lambda_i\right) \prod_{i=1}^m \lambda_i^{n-m} \prod_{i<j}^m (\lambda_i-\lambda_j)^2,
\end{align*}
where $\Gamma_p(\cdot)$ denotes the complex multivariate Gamma function.
\end{lem}

Let the eigenvalue decomposition of $\vH_1 \vH_1^*$ be $\vU {\bm\Lambda} \vU^*$. The distribution of eigenvalues $\mathrm{diag}(\bm\Lambda)$ comes directly from the lemma.
Define $\tilde{\vH}_2 \in \mathbb{C}^{r\times 2}$ as
$$
\tilde{\vH}_2 = \vU^* \vH_2 = [\vh_a, \vh_b].
$$
Because $\vH_1$ and $\vH_2$ are independent,
so is the unitary $\vU$ and $\vH_2$, and therefore $\tilde{\vH}_2$ has the same distribution as $\vH_2$.
We have the following equations:
\begin{align}
  \nonumber
  \left|\vI_r+\vH\vQ\vH^*\right| 
  \nonumber
  &= \left|\vI_r + \vH_1 \vQ_1 \vH_1^* + \vH_2 \vQ_2 \vH_2^* \right|\\
  \nonumber
  &= \left|\vI_r + q_0 \vU{\bm\Lambda} \vU^* + \vH_2 \vQ_2 \vH_2^* \right| \\
  \nonumber
  &= \left|{\bm\Lambda}' + \tilde{\vH}_2 \vQ_2 \tilde{\vH}_2^* \right| \\
  \nonumber
  &= \left|{\bm\Lambda}'\right| \left|\vI_2 + \vQ_2 \tilde{\vH}_2^* ({\bm\Lambda}')^{-1} \tilde{\vH}_2\right| \\
  \label{eqn_complex}
  &= \left(\prod_{i=1}^{r}\lambda'_i\right) \left( (1+q_a m_a)(1+q_b m_b)-q_a q_b \left|\xi_{ab}\right|^2 \right)
\end{align}
where ${\bm\Lambda}' = \vI_r + q_0 {\bm\Lambda}$ is diagonal, and
\begin{align}
  m_a &= \vh_a^*({\bm\Lambda}')^{-1}\vh_a, \\
  m_b &= \vh_b^*({\bm\Lambda}')^{-1}\vh_b, \\
  \xi_{ab} &= \vh_a^*({\bm\Lambda}')^{-1}\vh_b = \left(\vh_b^*({\bm\Lambda}')^{-1}\vh_a\right)^*.
\end{align}

The joint distribution density $\psi(m_a,m_b,\xi_{ab})$ can be directly derived from the distribution of standard complex Gaussian vectors $\vh_a, \vh_b$, and Lemma \ref{lemma1}. Denote \eqref{eqn_complex} as $F(\bm{\lambda},m_a,m_b,\xi_{ab})$. Then the outage probability can be calculated as
\begin{equation} \label{eqn_integralcomplex}
  \int_{F<{\rm e}^R-1}  \psi(m_a,m_b,\xi_{ab})  \mathrm{d}\bm{\lambda} \mathrm{d}m_a \mathrm{d}m_b \mathrm{d}\xi_{ab},
\end{equation}
This expression contains far less integral variables due to our simplification assumption.

This approach is also useful in that it can be directly applied to the calculation of derivatives. Recall Theorem \ref{thm_2} where our intuition comes from modifying two entries of $\vq$ at once, as can be seen from the proof. Therefore our specific type of $\vQ$ discussed here is suitable for the derivative calculation problem posed by Theorem \ref{thm_2}.

\begin{rem}
As an important note, our conversion of the original matrix-form inequality into the integral expression is crucial, because direct Monte Carlo simulation in matrix form can be highly unstable and therefore unreliable. Simplifying matrix expression into integral of elementary functions guarantees the accuracy of our results.
\end{rem}

\section{Numerical Verification}

In this section, we present some experiments to illustrate in what cases the Outage Probability Conjecture is likely to fail. Throughout our discussion, the number of transmitters $t$ is set as 2.

We take the system parameters $(r,R,P)$ as different sets of values, and look for the $q_m \in \left[ 0,{P}/{2} \right]$ that minimizes the outage probability for $t=2$ from the result of \eqref{eqn_solution} and \eqref{eqn_solution_0}. The number of receivers $r$ is set as 2.
We first conduct a broad search, where the channel rate $R$ ranges from $0.05r$ to $5r$ at step $0.05r$, and the power $P$ range from $0.05r$ to $5r$ at step $0.05r$.
The multiplication by $r$ is due to the consideration of controlling the rate and power per receiver.
Then we look for the $q_m$ that achieves the smallest outage probability among these sample values.
When the minimum is not attained at either end, it is sufficient to conclude that the conjecture is false.
The result is shown as the first figure in the top plot in Fig. \ref{fig_simu}, which only identifies one region where the conjecture fails.
A thick red dot indicates a $(R,P)$ set where we have sufficient evidence that the conjecture failed, and a blue dot indicates where the conjecture appears correct at least from our sample points.
The blank area on the bottom-right is discarded, where outage probability becomes so close to $1$ that numerical calculation becomes unstable.
Such cases are also meaningless because the channel would be useless in reality.

\begin{figure}[!t]
  \centering
  \includegraphics[width=2.8in]{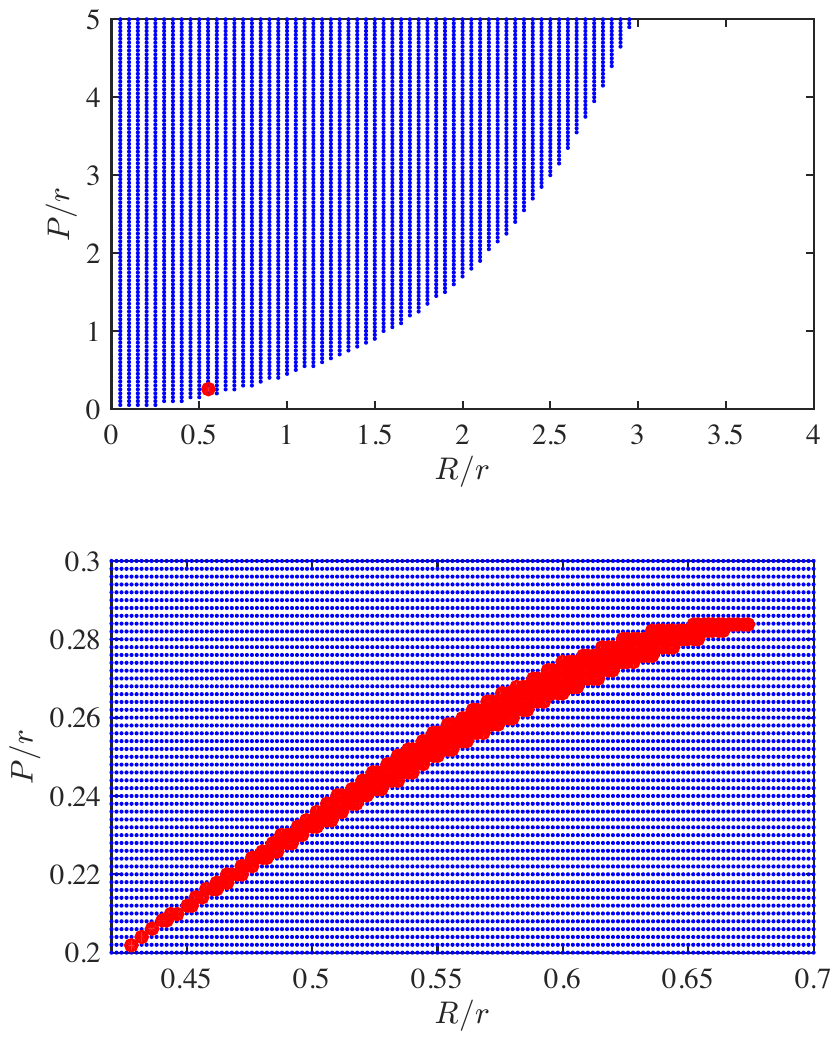}
  \caption{Test of the Conjecture with $t=r=2$ and $(R,P)$ varying. The Conjecture is mostly true with few exceptions. The bottom plot specifies the exception region in the top plot, where small blue dot and big red dot denotes the true and false sample, respectively.}
  \label{fig_simu}
\end{figure}

Taking a closer look at the failure region, we zoom in the parameters, and let $R$ range from $0.42r$ to $0.7r$ at step $0.02r$, and $P$ from $0.2r$ to $0.28r$ at step $0.02r$. $q$ is set at sample values from $0$ to $0.5P$ at step $0.025P$. The detail is plotted in the bottom plot in Fig. \ref{fig_simu}.

From the experiment result above, we find that the Outage Probability Conjecture is likely to be true for most parameter sets, which explains why it has been trusted by researchers for so many years.
We also tested the Conjecture for $r=3$ and $r=4$ with a broad range of parameters $(R,P)$, and surprisingly did not find any counter-example.
Future work could go deeper into this problem as for whether the identified region above is the only failure region, or with what parameters the Conjecture would fail.

\begin{rem}
However, we note that of all counter-examples we have identified so far, the actual outage probability is approximately 0.9. A channel with such a failure rate would be way too unreliable for any practical application. We suggest that the conjecture is still safe to use in any real-world circumstances where the outage probability is small, although this is still to be verified.
\end{rem}

\section{Conclusion}

In this paper, we present the theoretical analysis of the Outage Probability Conjecture regarding non-ergodic MIMO channel. A sufficient condition for counter-examples to the Conjecture is proposed, and a specific counter-example is shown that rejects the long-believed Conjecture. The counter-example is verified by both theoretical derivation and numerical experiment. We also propose an alternative method to evaluate the derivatives and the outage probability with particular power distributions for arbitrary $t$. Experiments are conducted to test the Conjecture on different parameter sets, revealing that the Conjecture is true in most cases, despite the counter-example identified here. Possible future work can address the problem of what particular parameters would be counter-examples of the Conjecture, and whether the Conjecture is true in almost every other case.

\section{Appendix}

We shall provide the derivation of the first and the second order derivatives that appear in Theorem \ref{thm_1}. The reason for only calculating derivatives up to the second order is explained in Remark \ref{rem_1}.

We first reiterate the fact that $f=f(q_1,q_2)$ is actually a univariate function with constraint $q_1+q_2=P$. We will use the derivative such as ${\mathrm{d}f}/{\mathrm{d}q_1}$ to denote the derivative of $f$ as a univariate function with the constraint, and use the partial derivative such as ${\partial f(q_1,q_2)}/{\partial q_1}$ to denote the partial derivative of $f$ as a bivariate function without any constraint. Such notation would facilitate our discussion.

For better checking, the partial derivative equations with the symmetry in this problem are copied below.
\begin{align*}
\frac{\mathrm{d} f}{\mathrm{d} q_1} &= \frac{\partial f(q_1, q_2)}{\partial q_1} - \frac{\partial f(q_1, q_2)}{\partial q_2} \\
&= \frac{\partial f(q_2, q_1)}{\partial q_1} - \frac{\partial f(q_1, q_2)}{\partial q_2},\\
\frac{\mathrm{d}^2 f}{\mathrm{d} q_1^2} &= \frac{\partial^2 f(q_1, q_2)}{\partial q_1^2} + \frac{\partial^2 f(q_1, q_2)}{\partial q_2^2} - 2\frac{\partial^2 f(q_1, q_2)}{\partial q_1\partial q_2} \\
&= \frac{\partial^2 f(q_2, q_1)}{\partial q_1^2} + \frac{\partial^2 f(q_1, q_2)}{\partial q_2^2} - 2\frac{\partial^2 f(q_1, q_2)}{\partial q_1\partial q_2}.
\end{align*}

Let us begin from calculating the first order derivatives. Denote $u={\rm e}^R-1$ for briefness. When $q_1, q_2 \ne 0$,
\begin{equation}\label{eqn_d1}
\frac{\partial f(q_1, q_2)}{\partial q_2}
= \int_0^1 (r-1)\rho^{r-2}{\rm d}\rho \cdot A,
\end{equation}
where
\begin{align}
A := &\int_0^u \frac{s^{r-1}{\rm e}^{-\frac{s}{q_1}}}{(r-1)!q_1^r}\cdot
\frac{\partial \int_0^{\frac{1}{q_2}\frac{u -s}{1+\rho s}} \frac{t^{r-1}{\rm e}^{-t}}{(r-1)!}{\rm d}t}{\partial q_2}{\rm d}s  \nonumber\\
=& - \int_0^u \frac{s^{r-1}\left(\frac{u -s}{1+\rho s}\right)^r}{(r-1)!^2q_1^rq_2^{r+1}}{\rm e}^{-\frac{s}{q_1}-\frac{1}{q_2}\frac{u -s}{1+\rho s}} {\rm d}s \nonumber\\
=& - \int_0^u \frac{s^r\left(\frac{u -s}{1+\rho s}\right)^{r-1}}{(r-1)!^2q_1^rq_2^{r+1}}\frac{1+\rho u}{(1+\rho t)^2}{\rm e}^{-\frac{s}{q_2}-\frac{1}{q_1}\frac{u -s}{1+\rho s}} {\rm d}s. \nonumber
\end{align}
Since $f(q_1, q_2) = f(q_2, q_1)$, we have
$$
\frac{\partial f(q_1, q_2)}{\partial q_1} = \frac{\partial f(q_2, q_1)}{\partial q_1}.
$$
Therefore \eqref{eqn_d1} can be used to calculate the first order derivatives at any point except for 0 and $P$.

In addition, when one of the variables is zero, we have
\begin{align}
\nonumber
\left. \frac{\partial f(q_1, q_2)}{\partial q_2}\right|_{\substack{q_1 = 0,\\ q_2 = P}} = -\frac{u^r{\rm e}^{-\frac{u}{P}}}{(r-1)!P^{r+1}}
\end{align}
and
\begin{align}
\nonumber
\left. \frac{\partial f(q_1, q_2)}{\partial q_2}\right|_{\substack{q_1 = P,\\ q_2 = 0}} = -\frac{u^{r-1}{\rm e}^{-\frac{u}{P}}}{(r-1)!P^r}(r + (r-1)u).
\end{align}

Next we calculate the second order derivatives.
When $q_1, q_2 \ne 0$, we have
\begin{equation}
\frac{\partial^2 f(q_1, q_2)}{\partial q_2^2}
= -\int_0^1 (r-1)\rho^{r-2}{\rm d}\rho \cdot B, \label{eqn_d21}
\end{equation}
where
\begin{align*}
B :=& \int_0^u \frac{\partial}{\partial q_2} \frac{s^{r-1}\left(\frac{u -s}{1+\rho s}\right)^r}{(r-1)!^2q_1^rq_2^{r+1}}{\rm e}^{-\frac{s}{q_1}-\frac{1}{q_2}\frac{u -s}{1+\rho s}}  {\rm d}s \\
=& - \int_0^u \frac{s^{r-1}\left(\frac{u -s}{1+\rho s}\right)^r}{(r-1)!^2q_1^rq_2^{r+2}}{\rm e}^{-\frac{s}{q_1}-\frac{1}{q_2}\frac{u -s}{1+\rho s}}
 \left(r+1-\frac{1}{q_2}\frac{u -s}{1+\rho s}\right) {\rm d}s,
\end{align*}
and
\begin{equation}
\frac{\partial^2 f(q_1, q_2)}{\partial q_1\partial q_2}
= -\int_0^1 (r-1)\rho^{r-2}{\rm d}\rho \cdot C, \label{eqn_d22}
\end{equation}
where
\begin{align*}
C :=& \int_0^u \frac{\partial}{\partial q_1} \frac{s^{r-1}\left(\frac{u -s}{1+\rho s}\right)^r}{(r-1)!^2q_1^rq_2^{r+1}}{\rm e}^{-\frac{s}{q_1}-\frac{1}{q_2}\frac{u -s}{1+\rho s}} {\rm d}s \\
=& - \int_0^u \frac{s^{r-1}\left(\frac{u -s}{1+\rho s}\right)^r}{(r-1)!^2q_1^{r+1}q_2^{r+1}}{\rm e}^{-\frac{s}{q_1}-\frac{1}{q_2}\frac{u -s}{1+\rho s}} \left(r-\frac{s}{q_1}\right) {\rm d}s.
\end{align*}
Equations \eqref{eqn_d21} and \eqref{eqn_d22} can be used to calculate the second order derivatives at any point except for 0 and $P$.

In addition, when one of the variables is zero, we have
\begin{align*}
\left.\frac{\partial^2 f(q_1, q_2)}{\partial q_2^2}\right|_{\substack{q_1 = 0,\\ q_2 = P}} =&
\frac{u^r{\rm e}^{-\frac{u}{P}}}{(r-1)!P^{r+2}}\left(r+1-\frac{u}{P}\right),\\
\left.\frac{\partial^2 f(q_1, q_2)}{\partial q_2^2}\right|_{\substack{q_1 = P\\ q_2 = 0}} =&
\frac{u^r{\rm e}^{-\frac{u}{P}}}{(r-1)!P^{r+1}}r(r+1)\\
&\quad\cdot\left(P(r-1)(\frac{1}{u}+1)^2-\frac{1}{u}-\frac{2(r-1)}{r}-\frac{u(r-1)}{r+1}\right),\\
\left.\frac{\partial^2 f(q_1, q_2)}{\partial q_1\partial q_2}\right|_{\substack{q_1 = 0,\\ q_2 = P}} =&
\frac{u^r{\rm e}^{-\frac{u}{P}}}{(r-1)!P^{r+2}}\left(\frac{rP}{u}-1\right)\left(r+(r-1)u\right),
\end{align*}
and
$$
\left. \frac{\partial^2 f(q_1, q_2)}{\partial q_1\partial q_2}\right|_{\substack{q_1 = P\\ q_2 = 0}} =
\left. \frac{\partial^2 f(q_1, q_2)}{\partial q_1\partial q_2}\right|_{\substack{q_1 = 0,\\ q_2 = P}}.
$$

\bibliographystyle{IEEEtran}
\bibliography{mybibfile}

\begin{thebibliography}{10}
\providecommand{\url}[1]{#1}
\csname url@samestyle\endcsname
\providecommand{\newblock}{\relax}
\providecommand{\bibinfo}[2]{#2}
\providecommand{\BIBentrySTDinterwordspacing}{\spaceskip=0pt\relax}
\providecommand{\BIBentryALTinterwordstretchfactor}{4}
\providecommand{\BIBentryALTinterwordspacing}{\spaceskip=\fontdimen2\font plus
\BIBentryALTinterwordstretchfactor\fontdimen3\font minus
  \fontdimen4\font\relax}
\providecommand{\BIBforeignlanguage}[2]{{%
\expandafter\ifx\csname l@#1\endcsname\relax
\typeout{** WARNING: IEEEtran.bst: No hyphenation pattern has been}%
\typeout{** loaded for the language `#1'. Using the pattern for}%
\typeout{** the default language instead.}%
\else
\language=\csname l@#1\endcsname
\fi
#2}}
\providecommand{\BIBdecl}{\relax}
\BIBdecl

\bibitem{telatar1999capacity}
E.~Telatar, ``Capacity of multi-antenna gaussian channels,'' \emph{Transactions
  on Emerging Telecommunications Technologies}, vol.~10, no.~6, pp. 585--595,
  1999.

\bibitem{paulraj1994increasing}
A.~J. Paulraj and T.~Kailath, ``Increasing capacity in wireless broadcast
  systems using distributed transmission/directional reception (dtdr),'' 1994,
  uS Patent 5,345,599.

\bibitem{raleigh1998spatio}
G.~G. Raleigh and J.~M. Cioffi, ``Spatio-temporal coding for wireless
  communication,'' \emph{IEEE Transactions on communications}, vol.~46, no.~3,
  pp. 357--366, 1998.

\bibitem{foschini1996layered}
G.~J. Foschini, ``Layered space-time architecture for wireless communication in
  a fading environment when using multi-element antennas,'' \emph{Bell labs
  technical journal}, vol.~1, no.~2, pp. 41--59, 1996.

\bibitem{zheng2003diversity}
L.~Zheng and D.~N.~C. Tse, ``Diversity and multiplexing: A fundamental tradeoff
  in multiple-antenna channels,'' \emph{IEEE Transactions on information
  theory}, vol.~49, no.~5, pp. 1073--1096, 2003.

\bibitem{andrews2007fundamentals}
J.~G. Andrews, A.~Ghosh, and R.~Muhamed, \emph{Fundamentals of WiMAX:
  understanding broadband wireless networking}.\hskip 1em plus 0.5em minus
  0.4em\relax Pearson Education, 2007.

\bibitem{xiao2005ieee}
Y.~Xiao, ``Ieee 802.11 n: enhancements for higher throughput in wireless
  lans,'' \emph{IEEE Wireless Communications}, vol.~12, no.~6, pp. 82--91,
  2005.

\bibitem{damnjanovic2011survey}
A.~Damnjanovic, J.~Montojo, Y.~Wei, T.~Ji, T.~Luo, M.~Vajapeyam, T.~Yoo,
  O.~Song, and D.~Malladi, ``A survey on 3gpp heterogeneous networks,''
  \emph{IEEE Wireless Communications}, vol.~18, no.~3, 2011.

\bibitem{boccardi2014five}
F.~Boccardi, R.~W. Heath, A.~Lozano, and T.~L. Marzetta, ``Five disruptive
  technology directions for 5g,'' \emph{Communications Magazine IEEE}, vol.~52,
  no.~2, pp. 74--80, 2014.

\bibitem{jungnickel2014role}
V.~Jungnickel, K.~Manolakis, W.~Zirwas, B.~Panzner, V.~Braun, M.~Lossow,
  M.~Sternad, R.~ApelfröJd, and T.~Svensson, ``The role of small cells,
  coordinated multipoint, and massive mimo in 5g,'' \emph{IEEE Communications
  Magazine}, vol.~52, no.~5, pp. 44--51, 2014.

\bibitem{ayach2014spatially}
O.~E. Ayach, S.~Rajagopal, S.~Abu-Surra, Z.~Pi, and R.~W. Heath, ``Spatially
  sparse precoding in millimeter wave mimo systems,'' \emph{IEEE Transactions
  on Wireless Communications}, vol.~13, no.~3, pp. 1499--1513, 2014.

\bibitem{heath2015overview}
R.~W. Heath, N.~González-Prelcic, S.~Rangan, W.~Roh, and A.~M. Sayeed, ``An
  overview of signal processing techniques for millimeter wave mimo systems,''
  \emph{IEEE Journal of Selected Topics in Signal Processing}, vol.~10, no.~3,
  pp. 436--453, 2015.

\bibitem{bj2015massive}
E.~Björnson, E.~G. Larsson, and T.~L. Marzetta, ``Massive mimo: ten myths and
  one critical question,'' \emph{IEEE Communications Magazine}, vol.~54, no.~2,
  pp. 114--123, 2015.

\bibitem{marzetta1999capacity}
T.~L. Marzetta and B.~M. Hochwald, ``Capacity of a mobile multiple-antenna
  communication link in rayleigh flat fading,'' \emph{IEEE transactions on
  Information Theory}, vol.~45, no.~1, pp. 139--157, 1999.

\bibitem{loyka2001channel}
S.~L. Loyka, ``Channel capacity of mimo architecture using the exponential
  correlation matrix,'' \emph{IEEE Communications letters}, vol.~5, no.~9, pp.
  369--371, 2001.

\bibitem{blum2003mimo}
R.~S. Blum, ``Mimo capacity with interference,'' \emph{IEEE Journal on selected
  areas in communications}, vol.~21, no.~5, pp. 793--801, 2003.

\bibitem{oyman2006non}
O.~Oyman and S.~Sandhu, ``Non-ergodic power-bandwidth tradeoff in linear
  multi-hop networks,'' in \emph{Information Theory, 2006 IEEE International
  Symposium on}.\hskip 1em plus 0.5em minus 0.4em\relax IEEE, 2006, pp.
  1514--1518.

\bibitem{letzepis2009outage}
N.~Letzepis and A.~G.~I. Fabregas, ``Outage probability of the gaussian mimo
  free-space optical channel with ppm,'' \emph{IEEE Transactions on
  Communications}, vol.~57, no.~12, 2009.

\bibitem{prasad2003outage}
N.~Prasad and M.~K. Varanasi, ``Outage analysis and optimization for
  multiaccess/v-blast architecture over mimo rayleigh fading channels,'' in
  \emph{PROCEEDINGS OF THE ANNUAL ALLERTON CONFERENCE ON COMMUNICATION CONTROL
  AND COMPUTING}, vol.~41, no.~1.\hskip 1em plus 0.5em minus 0.4em\relax The
  University; 1998, 2003, pp. 358--367.

\bibitem{abbe2012proof}
E.~Abbe, S.~L. Huang, and E.~Telatar, ``Proof of the outage probability
  conjecture for miso channels,'' in \emph{Information Theory Workshop, 2012
  IEEE}.\hskip 1em plus 0.5em minus 0.4em\relax IEEE, 2012, pp. 65--69.

\bibitem{harville1997matrix}
D.~A. Harville, ``Matrix algebra from a statistician's perspective,''
  \emph{Technometrics}, vol.~40, no.~2, 1997.

\bibitem{ratnarajah2005eigenvalues}
T.~Ratnarajah, R.~Vaillancourt, and M.~Alvo, ``Eigenvalues and condition
  numbers of complex random matrices,'' pp. 441--456, 2005.

\end{thebibliography}


\end{document}